\documentclass[12pt]{article}

\usepackage{amssymb}
\usepackage{amsmath}
\usepackage{graphicx}
\usepackage{hyperref}
\usepackage{tikz}
\usepackage{listings}
\usepackage{pythontex}
\usepackage{amsthm}

\newtheorem{Definition}{Definition}[section]
\newtheorem{Theorem}[Definition]{Theorem}
\newtheorem{Lemma}[Definition]{Lemma}

\newtheorem{Corollary}[Definition]{Corollary}

\newtheorem{conjecture}[Definition]{Conjecture}

\usepackage{array,multirow,makecell}

\setcellgapes{1pt}
\makegapedcells
\newcolumntype{R}[1]{>{\raggedleft\arraybackslash }b{#1}}
\newcolumntype{L}[1]{>{\raggedright\arraybackslash }b{#1}}
\newcolumntype{C}[1]{>{\centering\arraybackslash }b{#1}}
\newcounter{minutes}
\divide\time by 60
\newcounter{hours}
\multiply\time by 60 \addtocounter{minutes}{-\time}

\definecolor{codeblue}{rgb}{0.25, 0.5, 0.75}
\definecolor{codegreen}{rgb}{0, 0.6, 0}
\definecolor{codegray}{rgb}{0.5, 0.5, 0.5}
\definecolor{codepurple}{rgb}{0.58, 0, 0.82}
\definecolor{backcolour}{rgb}{0.95, 0.95, 0.92}

\makeatletter
\def\ps@pprintTitle{%
 \let\@oddhead\@empty
 \let\@evenhead\@empty
 \let\@oddfoot\@empty
 \let\@evenfoot\@oddfoot
}
\makeatother

\title{Linear Codes Derived from the Structure of Unit Graphs Over $\mathbb{Z}_n$}
\author{Apurba Sarkar$^1$,  Kalyan Hansda$^2$ and Makhan Maji$^3$ \\
\footnotesize{$^1,^2$Department of Mathematics, Visva-Bharati,}\\
\footnotesize{Santiniketan, Bolpur - 731235, West Bengal, India}\\
\footnotesize{$^3$Department of Mathematics, IIT Madras}\\
\footnotesize{Sardar Patel Road, Chennai, TN - 600036, India}\\
\footnotesize{apurbasarkar065@gmail.com$^1$}, \footnotesize{kalyanh4@gmail.com$^2$} and \footnotesize{makhan2maths@gmail.com$^3$}}
\begin{document}
\maketitle

\begin{abstract}
 In this paper, we study the unit graph $ G(\mathbb{Z}_n) $, where $ n $ is of the form  $n = p_1^{n_1} p_2^{n_2} \dots p_r^{n_r}$, with $ p_1, p_2, \dots, p_r $ being distinct prime numbers and $ n_1, n_2, \dots, n_r $ being positive integers. We establish the connectivity of $ G(\mathbb{Z}_n) $, show that its diameter is at most three, and analyze its edge connectivity. Furthermore, we construct $ q $-ary linear codes from the incidence matrix of $ G(\mathbb{Z}_n) $, explicitly determining their parameters and duals. A primary contribution of this work is the resolution of two conjectures from \cite{Jain2023} concerning the structural and coding-theoretic properties of $ G(\mathbb{Z}_n) $. These results extend the study of algebraic graph structures and highlight the interplay between number theory, graph theory, and coding theory.
\end{abstract}


\section{Introduction and Preliminaries}

Graph theory and algebraic structures share a profound interplay that has led to numerous advancements in both pure and applied mathematics. One of the intriguing graph-theoretic representations of algebraic structures is the \emph{unit graph} of a commutative ring, which encodes fundamental ring-theoretic properties into a graph-theoretic framework. The study of unit graphs provides valuable insights into algebraic graph theory, with applications in network security, cryptography, and coding theory.

Let $ R $ be a commutative ring with unity. The \emph{unit graph} $ G(R) $, introduced in \cite{Annamalai2021}, is a simple graph whose vertex set consists of the elements of $ R $, and two distinct vertices $ a $ and $ b $ are adjacent if their sum is a unit, i.e., $ a + b \in U(R) $, where $ U(R) $ denotes the set of units in $ R $. The set of non-units is denoted by $ N_U(R) $. This graphical representation captures algebraic properties of rings and has been widely studied in recent years.

One of the fundamental parameters of a graph is its \emph{diameter}, which measures the maximum shortest-path distance between any two vertices. For a given graph $ G $, the diameter is defined as  $\text{diam}(G) = \max\{d(a, b) \mid a, b \in V(G)\}$, where $ d(a, b) $ denotes the shortest path between $ a $ and $ b $ in $ G $ \cite{Clark1991}. A graph's girth, represented by $g_r(G)$, is the length of its shortest cycle. Another important measure is the \emph{edge connectivity} $ \lambda(G) $, which is the minimum number of edges that must be removed to disconnect $ G $ or reduce it to a trivial structure \cite{Clark1991}. Chartrand \cite{Chartrand1966} established that if a connected graph satisfies $ \text{diam}(G) \leq 2 $, then $ \lambda(G) = \delta(G) $, where $ \delta(G) $ is the minimum degree of $ G $. Similarly, Plesník and Znám \cite{Plesnik1989} proved that for a connected bipartite graph with $ \text{diam}(G) \leq 3 $, the same result holds. These properties are crucial in studying the structure and connectivity of unit graphs.

Beyond its graph-theoretic significance, the unit graph has found applications in the construction of \emph{linear codes}, which are fundamental in error correction, cryptography, and secure communications. A linear code $ C $ over a finite field $ \mathbb{F}_q $ of length $ n $ is a subspace of $ \mathbb{F}_q^n $ \cite{LingXing2004}. The key parameters of $ C $ include its dimension $ k = \dim(C) $ and minimum Hamming distance $ d(C) $, which determine its error-detection and error-correction capabilities. The \emph{dual code} $ C^\perp $ consists of all vectors in $ \mathbb{F}_q^n $ that are orthogonal to every codeword in $ C $, with dimension satisfying $ \dim(C^\perp) = n - \dim(C) $. The minimum Hamming distance is defined as  $d(C) = \min\{d_H(a, b) \mid a, b \in C, a \neq b\}$, where $ d_H(a, b) $ is the number of positions at which $ a $ and $ b $ differ. A linear code with length $ n $, dimension $ k $, and minimum distance $ d $ is denoted by $ [n, k, d] $. The generator matrix of $ C $ is a matrix whose rows form a basis for $ C $, and the generator matrix of dual code $C^{\perp}$ is known as the parity-check matrix \cite{LingXing2004}.

The connection between unit graphs and coding theory has led to significant research advancements. Annamalai and Durairajan \cite{Annamalai2021} constructed linear codes from the incidence matrices of unit graphs $ G(\mathbb{Z}_p) $ and $ G(\mathbb{Z}_{2p}) $, where $ p $ is an odd prime. Jain, Reddy, and Shaikh \cite{Jain2023, Jain2024} extended this study by constructing $ q $-ary linear codes from incidence matrices of $ G(\mathbb{Z}_n) $ when $ n $ is a product of powers of three distinct primes, determining their parameters, and exploring decoding techniques. More recently, Jain, Reddy and Shaikh \cite{Jain2024} generalized these results to the $ r $ ary codes derived from $ G(\mathbb{Z}_n \oplus \mathbb{Z}_m) $. These studies led to two conjectures in \cite{Jain2023} regarding the structure of unit graphs and the characterization of their associated linear codes.

Motivated by these developments, we extend this line of research by analyzing the construction of linear codes from unit graphs associated with
$G(\mathbb{Z}_{n})$ for any positive integer $n$. We establish fundamental graph-theoretic properties of these unit graphs, proving that they are connected and have a diameter of at most three. Using their incidence matrices, we construct $ q $-ary linear codes, explicitly determining their parameters and investigating their duals. Furthermore, we resolve two conjectures from \cite{Jain2023} concerning the connectivity and coding-theoretic properties of unit graphs.

The results in this paper contribute to a deeper understanding of algebraic graph structures and their applications in coding theory and discrete mathematics. To support our analysis, we begin with the following useful lemma.

\begin{Lemma} \label{unit+non-unit}
In the ring $ \mathbb{Z}_{p^n} $, the sum of a unit and a non-unit is always a unit.
\end{Lemma}

The following result is due to T. A. de Lima, A. L. Galdino, A. B. Avelar \& M. A. Rincón \cite{Lima2021} and Proposition 1.10 of \cite{AtiyahMacdonald1994}.

\begin{Theorem}[Chinese Remainder Theorem for $n$-ideals]\label{CRT}
Let $ R $ be a ring with unity, and let $ A_1, A_2, \dots, A_k $ be ideals of $ R $. Then the following statements hold:
\begin{itemize}
    \item [i.] The map
    $$\phi: R \to R/A_1 \times R/A_2 \times \dots \times R/A_k$$
    defined by $ \phi(r) = (r + A_1, r + A_2, \dots, r + A_k) $ is a ring homomorphism with kernel
    $$\ker(\phi) = A_1 \cap A_2 \cap \dots \cap A_k.$$

    \item [ii.] If the ideals satisfy the comaximality condition
    $$A_i + A_j = R, \quad \text{for all } i, j \in \{1, 2, \dots, k\}, \; i \neq j,$$
    then  $R / (A_1 \cap A_2 \cap \dots \cap A_k) \cong R/A_1 \times R/A_2 \times \dots \times R/A_k$. In other words,  $R / (A_1 \cap A_2 \cap \dots \cap A_k) \cong R/A_1 \oplus R/A_2 \oplus  \dots\oplus  R/A_k$.
\end{itemize}
\end{Theorem}
\section{Main Results}
In this section, we establish key properties of unit graphs and construct linear codes from their incidence matrices. We begin by stating two conjectures concerning the structure and coding-theoretic properties of unit graphs.
\begin{conjecture}\cite{Jain2023}
Let $G(\mathbb{Z}_n)$ be a unit graph where n is any natural number. Then $G(\mathbb{Z}_n)$ is a connected graph, and
\begin{enumerate}
    \item If $2\in U(\mathbb{Z}_n)$, then $\text{diam}(G(\mathbb{Z}_n))\leq 2$.
    \item If $2\in N_U(\mathbb{Z}_n)$, then $\text{diam}(G(\mathbb{Z}_n))\leq 3$.
\end{enumerate}
\end{conjecture}
\begin{conjecture}\cite{Jain2023}
Let $G(\mathbb{Z}_n)$ be a unit graph and $H$ be a $|V|\times|E|$ incidence matrix of $G(\mathbb{Z}_n)$ where n is any natural number. Then:
\begin{enumerate}
    \item If $2\in U(\mathbb{Z}_n)$, then the binary code generated by $H$ is a
    $$C_2(H)=\left[\frac{(n-1)\phi(n)}{2}, n-1, \phi(n)-1\right]_2$$
    code over the finite field $\mathbb{F}_2$.
    \item If $2\in N_U(\mathbb{Z}_n)$, then for any odd prime $q$, the $q$-ary code generated by $H$ is a
    $$C_q(H)=\left[\frac{n\phi(n)}{2}, n-1, \phi(n)\right]_q$$
    code over the finite field $\mathbb{F}_q$.
\end{enumerate}
\end{conjecture}
For any $ n \in \mathbb{N} $, by Theorem~\ref{CRT},
if $n = p_1^{n_1} p_2^{n_2} \dots p_r^{n_r}$,
 where $ p_1, p_2, \dots, p_r $ are distinct prime numbers and $ n_1, n_2, \dots, n_r $ are positive integers, then the ring of integers modulo $ n $, denoted by $ \mathbb{Z}_n $, is isomorphic to the direct sum of the rings $ \mathbb{Z}_{p_1^{n_1}}, \mathbb{Z}_{p_2^{n_2}}, \dots, \mathbb{Z}_{p_r^{n_r}} $. More precisely,
 $$\mathbb{Z}_n \cong \mathbb{Z}_{p_1^{n_1}} \oplus \mathbb{Z}_{p_2^{n_2}} \oplus \dots \oplus \mathbb{Z}_{p_r^{n_r}}.$$
We now proceed to prove these conjectures and analyze their implications for unit graphs and their associated linear codes.

\subsection{Characterization of the unit graph $G(\mathbb{Z}_{p_{1}^{n_{1}}} \oplus \mathbb{Z}_{p_{2}^{n_{2}}} \oplus \cdots \oplus \mathbb{Z}_{p_{r}^{n_{r}}})$:\\}

This  section explores some key characterizations of the unit graph $G(\mathbb{Z}_{p_{1}^{n_{1}}} \oplus \mathbb{Z}_{p_{2}^{n_{2}}} \oplus \cdots \oplus \mathbb{Z}_{p_{r}^{n_{r}}})$ over the ring $\mathbb{Z}_{p_{1}^{n_{1}}} \oplus \mathbb{Z}_{p_{2}^{n_{2}}} \oplus \cdots \oplus \mathbb{Z}_{p_{r}^{n_{r}}}$. Throughout this section and following sections,  we denote the subsets $U(R)$ and $N_U (R)$ of the ring with unity $R$ as the units and non-units of $R$ respectively.

\begin{Theorem} \label{Number of edges}
The number of edges in the unit graph $G(\mathbb{Z}_{n_1} \oplus \mathbb{Z}_{n_2} \oplus \ldots \oplus \mathbb{Z}_{n_r})$, where $n_1,n_2,\ldots,n_r$ are
any positive integers is $\frac{1}{2}(n_1 n_2 \ldots n_r-1)\phi(n_1)\phi(n_2)\ldots\phi(n_r)$ if all $n_i$ are odd, and $\frac{1}{2}n_1 n_2 \ldots n_r\phi(n_1)\phi(n_2)\ldots\phi(n_r)$ otherwise.
\end{Theorem}

\begin{Theorem} \label{bipartite}
Let $G(\mathbb{Z}_{n_1} \oplus \mathbb{Z}_{n_2} \oplus \ldots \oplus
\mathbb{Z}_{n_r})$ be a
unit graph, where $n_1,n_2,\ldots,n_r$ are any positive integers. If
exactly one of $n_1,n_2,\ldots,n_r$ is even, then $G(\mathbb{Z}_{n_1} \oplus \mathbb{Z}_{n_2} \oplus \ldots \oplus
\mathbb{Z}_{n_r})$ is bipartite.
\end{Theorem}

\begin{proof}
Assuming that $ n_r $ is even and all $n_1,n_2,\ldots,n_{r-1}$ are odd,
let us define the sets $ V_1 $ and $ V_2 $ as follows:
$$V_1 = \{(a_1, a_2, \ldots, a_r) \in \mathbb{Z}_{n_1} \oplus \mathbb{Z}_{n_2} \oplus \ldots \oplus \mathbb{Z}_{n_r} \mid a_r = 2\beta, \beta \in \mathbb{Z} \},$$
and
$$V_2 = \{(a_1, a_2, \ldots, a_r) \in \mathbb{Z}_{n_1} \oplus \mathbb{Z}_{n_2} \oplus \ldots \oplus \mathbb{Z}_{n_r} \mid a_r = 2\beta + 1, \beta \in \mathbb{Z} \}.$$
These sets partition $ \mathbb{Z}_{n_1} \oplus \mathbb{Z}_{n_2} \oplus \ldots \oplus
\mathbb{Z}_{n_r} $ in such a way that no vertex in $ V_1 $ is
adjacent to any vertex in $ V_1 $, and the same holds for $ V_2 $.
\end{proof}

\begin{Theorem} \label{leq2}
The following statements hold in the unit graph $G(\mathbb{Z}_{p_{1}^{n_{1}}} \oplus \mathbb{Z}_{p_{2}^{n_{2}}} \oplus \cdots \oplus \mathbb{Z}_{p_{r}^{n_{r}}})$, for $p_1$, $p_2$, $\ldots$, $p_r$ all being odd primes.
\begin{enumerate}
\item $G(\mathbb{Z}_{p_{1}^{n_{1}}} \oplus \mathbb{Z}_{p_{2}^{n_{2}}} \oplus \cdots \oplus \mathbb{Z}_{p_{r}^{n_{r}}})$ is connected.
\item $diam(G(\mathbb{Z}_{p_{1}^{n_{1}}} \oplus \mathbb{Z}_{p_{2}^{n_{2}}} \oplus \cdots \oplus \mathbb{Z}_{p_{r}^{n_{r}}}))\leq2$.
\item $\lambda(G(\mathbb{Z}_{p_{1}^{n_{1}}} \oplus
\mathbb{Z}_{p_{2}^{n_{2}}} \oplus \cdots \oplus
\mathbb{Z}_{p_{r}^{n_{r}}}))= \phi(p_1^{n_1}) \phi( p_2^{n_2})\ldots
\phi(p_r^{n_r})-1$.
\end{enumerate}
\end{Theorem}

\begin{proof}
(1) Let $\mathbb{Z}_{p_{1}^{n_{1}}} \oplus \mathbb{Z}_{p_{2}^{n_{2}}} \oplus \cdots \oplus \mathbb{Z}_{p_{r}^{n_{r}}}= U(\mathbb{Z}_{p_{1}^{n_{1}}} \oplus \mathbb{Z}_{p_{2}^{n_{2}}} \oplus \cdots \oplus \mathbb{Z}_{p_{r}^{n_{r}}}) \cup N_{U}(\mathbb{Z}_{p_{1}^{n_{1}}} \oplus \mathbb{Z}_{p_{2}^{n_{2}}} \oplus \cdots \oplus \mathbb{Z}_{p_{r}^{n_{r}}})$.
  Without loss of generality we consider an element $\bar{a}\in\mathbb{Z}_{p_{1}^{n_{1}}} \oplus \mathbb{Z}_{p_{2}^{n_{2}}} \oplus \cdots \oplus \mathbb{Z}_{p_{r}^{n_{r}}}$ as $\bar{a}=(\alpha_{1}p_{1},\ldots,\alpha_{i}p_{i},u_{i+1},\ldots,u_{i+j},\underbrace{0,\ldots,0}_{k-\textrm{times}})$, where $i+j+k=r$. Let $\bar{a}$, $\bar{b}$ $\in\mathbb{Z}_{p_{1}^{n_{1}}} \oplus \mathbb{Z}_{p_{2}^{n_{2}}} \oplus \cdots \oplus \mathbb{Z}_{p_{r}^{n_{r}}}$. Here, we consider the following cases:

  \textbf{Case 1 :} Let $\bar{a}$, $\bar{b}$ $\in U(\mathbb{Z}_{p_{1}^{n_{1}}} \oplus \mathbb{Z}_{p_{2}^{n_{2}}} \oplus \cdots \oplus \mathbb{Z}_{p_{r}^{n_{r}}})$. 

    \textbf{Case 2 :} Next suppose $\bar{a}\in U(\mathbb{Z}_{p_{1}^{n_{1}}} \oplus \mathbb{Z}_{p_{2}^{n_{2}}} \oplus \cdots \oplus \mathbb{Z}_{p_{r}^{n_{r}}})$ and $\bar{b}\in N_{U}(\mathbb{Z}_{p_{1}^{n_{1}}} \oplus \mathbb{Z}_{p_{2}^{n_{2}}} \oplus \cdots \oplus \mathbb{Z}_{p_{r}^{n_{r}}})$. 

   \textbf{ Case 3:} Finally suppose $\bar{a}$, $\bar{b}$ $\in N_{U}(\mathbb{Z}_{p_{1}^{n_{1}}} \oplus \mathbb{Z}_{p_{2}^{n_{2}}} \oplus \cdots \oplus \mathbb{Z}_{p_{r}^{n_{r}}})$. 

  If $\bar{a}$ and $\bar{b}$ are adjacent, then $d(\bar{a},\bar{b})=1$. If not, let us choose
  \begin{align*}
    \bar{c}=&(w_{1},\ldots,w_{i_1},v_{i_1+1},\ldots,v_{i},u_{i+1},\ldots,u_{i+j_1},\underbrace{0,\ldots,0}_{j_2-\textrm{times}},w_{i+j+1},\ldots,w_{i+j+k_1},\\
            &v_{i+j+k_1+1},\ldots,v_{i+j+k_1+k_2},w_{i+j+k_1+k_2+1},\ldots,w_{i+j+k}).
  \end{align*}

  From above discussion we get, $\bar{a}$ and $\bar{b}$ are either adjacent or there exist $\bar{c}$ such that $[\bar{a},\bar{c}]$, $[\bar{c},\bar{b}]$ $\in G(\mathbb{Z}_{p_{1}^{n_{1}}} \oplus \mathbb{Z}_{p_{2}^{n_{2}}} \oplus \cdots \oplus \mathbb{Z}_{p_{r}^{n_{r}}})$.
  
  $d(\bar{a},\bar{b})\leq 2$ for all $\bar{a}$, $\bar{b}$ $\in\mathbb{Z}_{p_{1}^{n_{1}}} \oplus \mathbb{Z}_{p_{2}^{n_{2}}} \oplus \cdots \oplus \mathbb{Z}_{p_{r}^{n_{r}}}$. Hence, $G(\mathbb{Z}_{p_{1}^{n_{1}}} \oplus \mathbb{Z}_{p_{2}^{n_{2}}} \oplus \cdots \oplus \mathbb{Z}_{p_{r}^{n_{r}}})$ is connected.

(2) From the proof of (1), we have $d(\bar{a},\bar{b})\leq 2$ for all $\bar{a}$, $\bar{b}$ $\in\mathbb{Z}_{p_{1}^{n_{1}}} \oplus \mathbb{Z}_{p_{2}^{n_{2}}} \oplus \cdots \oplus \mathbb{Z}_{p_{r}^{n_{r}}}$. Hence $\textrm{diam}(G(\mathbb{Z}_{p_{1}^{n_{1}}} \oplus \mathbb{Z}_{p_{2}^{n_{2}}} \oplus \cdots \oplus \mathbb{Z}_{p_{r}^{n_{r}}}))\leq 2$.

(3) The minimum degree of $G(\mathbb{Z}_{p_{1}^{n_{1}}} \oplus \mathbb{Z}_{p_{2}^{n_{2}}} \oplus \cdots \oplus \mathbb{Z}_{p_{r}^{n_{r}}})$ is $\phi(p_1^{n_1}) \phi( p_2^{n_2})\ldots
\phi(p_r^{n_r})-1$. From \cite{Chartrand1966} we get, for a connected graph with a diameter less or equal to $2$, edge connectivity of the graph is equal to the minimum degree of the graph. Therefore, $\lambda(G(\mathbb{Z}_{p_{1}^{n_{1}}} \oplus
\mathbb{Z}_{p_{2}^{n_{2}}} \oplus \cdots \oplus
\mathbb{Z}_{p_{r}^{n_{r}}}))= \phi(p_1^{n_1}) \phi( p_2^{n_2})\ldots
\phi(p_r^{n_r})-1$.
\end{proof}

The following result determines the parameters of a binary linear code, along with its dual, generated by the incidence matrix of the unit graph $G(\mathbb{Z}_{p_{1}^{n_{1}}} \oplus \mathbb{Z}_{p_{2}^{n_{2}}}
\oplus \cdots \oplus \mathbb{Z}_{p_{r}^{n_{r}}})$, where $p_1$, $p_2$, $\ldots$, $p_r$ are distinct odd primes.

\begin{Theorem}\label{binary code 1  }
Let $G(\mathbb{Z}_{p_{1}^{n_{1}}} \oplus \mathbb{Z}_{p_{2}^{n_{2}}} \oplus \cdots \oplus \mathbb{Z}_{p_{r}^{n_{r}}})$ be the unit graph and let $H$ be its incidence matrix of size $|V| \times |E|$. If $p_1$, $p_2$, $\ldots$, $p_r$ are all distinct odd prime numbers, then the binary linear code $C_2(H)$ generated by $H$ over the finite field $\mathbb{F}_2$ is a $[\frac{1}{2}(p_{1}^{n_{1}}p_{2}^{n_{2}}\ldots
p_{r}^{n_{r}}
-1)\phi(p_{1}^{n_{1}})\phi(p_{2}^{n_{2}})\ldots\phi(p_{r}^{n_{r}}),\;p_{1}^{n_{1}}p_{2}^{n_{2}}\ldots
p_{r}^{n_{r}}-1,\;\phi(p_{1}^{n_{1}})\phi(p_{2}^{n_{2}})\ldots\phi(p_{r}^{n_{r}})-1]_2$-linear code.
Moreover, the dual code $C_2(H)^\bot$ is a $[\frac{1}{2}(p_{1}^{n_{1}}p_{2}^{n_{2}}\ldots p_{r}^{n_{r}}
-1)\phi(p_{1}^{n_{1}})\phi(p_{2}^{n_{2}})\ldots\phi(p_{r}^{n_{r}}),\;\frac{1}{2}(p_{1}^{n_{1}}p_{2}^{n_{2}}\ldots p_{r}^{n_{r}}-1)(\phi(p_{1}^{n_{1}})\phi(p_{2}^{n_{2}})\ldots\phi(p_{r}^{n_{r}})-2),\;3]_2$-linear code.
\end{Theorem}

\begin{proof}
    By Theorem \ref{leq2}, the unit graph $G(\mathbb{Z}_{p_{1}^{n_{1}}} \oplus \mathbb{Z}_{p_{2}^{n_{2}}} \oplus \cdots \oplus \mathbb{Z}_{p_{r}^{n_{r}}})$ is connected. 

By Theorem 4.2.4(ii) of \cite{LingXing2004}, $\dim C_2(H)^\bot$ is $\frac{1}{2}(p_{1}^{n_{1}}p_{2}^{n_{2}}\ldots p_{r}^{n_{r}}-1)(\phi(p_{1}^{n_{1}})\phi(p_{2}^{n_{2}})\ldots\phi(p_{r}^{n_{r}})-2)$. Also from Theorem 6 of \cite{Dankelmann2013}, the minimum distance of the dual code is given by the  $ g_r(G(\mathbb{Z}_{p_{1}^{n_{1}}} \oplus \mathbb{Z}_{p_{2}^{n_{2}}} \oplus \cdots \oplus \mathbb{Z}_{p_{r}^{n_{r}}})) $.    Thus, we conclude that
$g_r(G(\mathbb{Z}_{p_{1}^{n_{1}}} \oplus \mathbb{Z}_{p_{2}^{n_{2}}} \oplus \cdots \oplus \mathbb{Z}_{p_{r}^{n_{r}}})) = 3$, which implies that the minimum distance of the dual code is $3$.
\end{proof}

\begin{Corollary}
The incidence matrix of order $ |V| \times |E| $ corresponding to the unit graph $G(\mathbb{Z}_{p_{1}^{n_{1}}} \oplus \mathbb{Z}_{p_{2}^{n_{2}}} \oplus \cdots \oplus \mathbb{Z}_{p_{r}^{n_{r}}}) = (V,E)$, where $ p_1, p_2, \ldots, p_r $ are distinct odd primes, generates a linear code whose dual exhibits both $ 2 $-error-detection and single-error-correction capabilities.
  \end{Corollary}

\begin{proof}
This is straightforward.
\end{proof}

Now, we aim to explore the linear code generated by the incidence matrix of the unit graph of type $G(\mathbb{Z}_{p_{1}^{n_{1}}} \oplus \cdots \oplus \mathbb{Z}_{p_{r}^{n_{r}}} \oplus \mathbb{Z}_{2^{m}})$. We start with the following theorem that studies some primary characterizations of the said graph.

\begin{Theorem} \label{leq3}
The following statements hold in the unit graph $G(\mathbb{Z}_{p_{1}^{n_{1}}} \oplus \cdots \oplus \mathbb{Z}_{p_{r}^{n_{r}}} \oplus \mathbb{Z}_{2^{m}})$, for $p_1$, $p_2$, $\ldots$, $p_r$ all being distinct odd primes.
\begin{enumerate}
\item $G(\mathbb{Z}_{p_{1}^{n_{1}}} \oplus \cdots \oplus \mathbb{Z}_{p_{r}^{n_{r}}} \oplus \mathbb{Z}_{2^{m}})$ is connected.
\item $diam(G(\mathbb{Z}_{p_{1}^{n_{1}}} \oplus \cdots \oplus \mathbb{Z}_{p_{r}^{n_{r}}} \oplus \mathbb{Z}_{2^{m}}))\leq3$.
\item $\lambda(G(\mathbb{Z}_{p_{1}^{n_{1}}} \oplus \cdots \oplus
\mathbb{Z}_{p_{r}^{n_{r}}} \oplus \mathbb{Z}_{2^{m}}))=2^{m-1}
\phi(p_1^{n_1}) \phi( p_2^{n_2})\ldots \phi(p_r^{n_r})$.
\end{enumerate}
\end{Theorem}

\begin{proof}
 $(1):$ Let $\mathbb{Z}_{p_{1}^{n_{1}}} \oplus \cdots \oplus \mathbb{Z}_{p_{r}^{n_{r}}} \oplus \mathbb{Z}_{2^{m}}= U(\mathbb{Z}_{p_{1}^{n_{1}}} \oplus \cdots \oplus \mathbb{Z}_{p_{r}^{n_{r}}} \oplus \mathbb{Z}_{2^{m}}) \cup N_{U}(\mathbb{Z}_{p_{1}^{n_{1}}} \oplus \cdots \oplus \mathbb{Z}_{p_{r}^{n_{r}}} \oplus \mathbb{Z}_{2^{m}} )$.
  Without loss of generality we consider an element $\bar{a}\in \mathbb{Z}_{p_{1}^{n_{1}}} \oplus \cdots \oplus \mathbb{Z}_{p_{r}^{n_{r}}} \oplus \mathbb{Z}_{2^{m}}$ as $\bar{a}=(\alpha_{1}p_{1},\ldots,\alpha_{i_1}p_{i_1},u_{i_1+1},\ldots,u_{i_1+i_2},\underbrace{0,\ldots,0}_{i_3-\textrm{times}},t)$, where $i_1+i_2+i_3=r$. Let $\bar{a}$, $\bar{b}$ $\in\mathbb{Z}_{p_{1}^{n_{1}}} \oplus \cdots \oplus \mathbb{Z}_{p_{r}^{n_{r}}} \oplus \mathbb{Z}_{2^{m}}$. Here, we consider the following cases:

  \textbf{Case 1 :} Let $\bar{a}$, $\bar{b}$ $\in U(\mathbb{Z}_{p_{1}^{n_{1}}} \oplus \cdots \oplus \mathbb{Z}_{p_{r}^{n_{r}}} \oplus \mathbb{Z}_{2^{m}})$. 
   \textbf{ Case 2 : }Next suppose $\bar{a}\in U(\mathbb{Z}_{p_{1}^{n_{1}}} \oplus \cdots \oplus \mathbb{Z}_{p_{r}^{n_{r}}} \oplus \mathbb{Z}_{2^{m}})$ and $\bar{b}\in N_{U}(\mathbb{Z}_{p_{1}^{n_{1}}} \oplus \cdots \oplus \mathbb{Z}_{p_{r}^{n_{r}}} \oplus \mathbb{Z}_{2^{m}})$. Then,
        \begin{align*}
        \bar{a}=&(u_{1},\ldots,u_{i_1},u_{i_1+1},\ldots,u_{i_1+i_2},u_{i_1+i_2+1},\ldots,u_{i_1+i_2+i_3},u'_{r+1}),\\
        \bar{b}=&(\alpha_{1}p_{1},\ldots,\alpha_{i_1}p_{i_1},v_{i_1+1},\ldots,v_{i_1+i_2},\underbrace{0,\ldots,0}_{i_3-\textrm{times}},t).
        \end{align*}
        If $\bar{a}$ and $\bar{b}$ are adjacent, then $d(\bar{a},\bar{b})=1$. If not, note that $t$ is either a unit or non-unit. 
              
       \textbf{ Case 3 :} Finally suppose, both $\bar{a}$, $\bar{b}$ $\in N_{U}(\mathbb{Z}_{p_{1}^{n_{1}}} \oplus \cdots \oplus \mathbb{Z}_{p_{r}^{n_{r}}} \oplus \mathbb{Z}_{2^{m}})$. 
            If $\bar{a}$ and $\bar{b}$ are adjacent, then $d(\bar{a},\bar{b})=1$. If not, note that $t$ and $t'$ are either unit or non-unit. Accordingly, we consider the following cases:

            \textbf{Subcase(a) :} First suppose, both $t$ and $t'$ are non-unit. 

              By Lemma~\ref{unit+non-unit} and using the fact that a product of two units is also a unit, $\bar{a}+\bar{c}$, $\bar{c}+\bar{b}$ $\in U(\mathbb{Z}_{p_{1}^{n_{1}}} \oplus \cdots \oplus \mathbb{Z}_{p_{r}^{n_{r}}} \oplus \mathbb{Z}_{2^{m}})$. Therefore, $[\bar{a},\bar{c}]$, $[\bar{c},\bar{b}]$ $\in G(\mathbb{Z}_{p_{1}^{n_{1}}} \oplus \cdots \oplus \mathbb{Z}_{p_{r}^{n_{r}}} \oplus \mathbb{Z}_{2^{m}})$. Hence $d(\bar{a},\bar{b})=2$.

         \textbf{Subcase(b) :} Next suppose that $t$ and $t'$ are unit and nonunit, respectively. So we choose \\ $\bar{c}=(u_{1},\ldots,u_{i_1},v_{i_1+1},\ldots,v_{i_1+i_2},u_{i_1+i_2+1},\ldots,u_{i_1+i_2+i_3},u'_{r+1})$ and \\ $\bar{d}=(u_{1},\ldots,u_{i_1},v_{i_1+1},\ldots,v_{i_1+i_2},\underbrace{0,\ldots,0}_{i_3-\textrm{times}},0)$.
              By Lemma~\ref{unit+non-unit} and using the fact that, product of two units is also a unit, $\bar{a}+\bar{c}$, $\bar{c}+\bar{d}$, $\bar{d}+\bar{b}$ $\in U(\mathbb{Z}_{p_{1}^{n_{1}}} \oplus \cdots \oplus \mathbb{Z}_{p_{r}^{n_{r}}} \oplus \mathbb{Z}_{2^{m}})$. Therefore, $[\bar{a},\bar{c}]$, $[\bar{c},\bar{d}]$, $[\bar{d},\bar{b}]$ $\in G(\mathbb{Z}_{p_{1}^{n_{1}}} \oplus \cdots \oplus \mathbb{Z}_{p_{r}^{n_{r}}} \oplus \mathbb{Z}_{2^{m}})$. Hence $d(\bar{a},\bar{b})=3$.

         \textbf{Subcase(c) :} Finally suppose, both $t$, $t'$ $\in U(\mathbb{Z}_{2^m})$. In this case we choose \\ $\bar{c}=(u_{1},\ldots,u_{i_1},v_{i_1+1},\ldots,v_{i_1+i_2},\underbrace{0,\ldots,0}_{i_3-\textrm{times}},0)$.

         Therefore, $d(\bar{a},\bar{b})\leq 3$.

  From above discussion we get, $d(\bar{a},\bar{b})\leq 3$ for all $\bar{a}$, $\bar{b}$ $\in\mathbb{Z}_{p_{1}^{n_{1}}} \oplus \cdots \oplus \mathbb{Z}_{p_{r}^{n_{r}}} \oplus \mathbb{Z}_{2^{m}}$. Hence, $G(\mathbb{Z}_{p_{1}^{n_{1}}} \oplus \cdots \oplus \mathbb{Z}_{p_{r}^{n_{r}}} \oplus \mathbb{Z}_{2^{m}})$ is connected.

(2) From the proof of (1), we have $d(\bar{a},\bar{b})\leq 3$ for all $\bar{a}$, $\bar{b}$ $\in \mathbb{Z}_{p_{1}^{n_{1}}} \oplus \cdots \oplus \mathbb{Z}_{p_{r}^{n_{r}}} \oplus \mathbb{Z}_{2^{m}}$. Hence $\textrm{diam}(G(\mathbb{Z}_{p_{1}^{n_{1}}} \oplus \cdots \oplus \mathbb{Z}_{p_{r}^{n_{r}}} \oplus \mathbb{Z}_{2^{m}}))\leq 3$.

(3) The minimum degree of the unit graph $G(\mathbb{Z}_{p_{1}^{n_{1}}} \oplus \cdots \oplus \mathbb{Z}_{p_{r}^{n_{r}}} \oplus \mathbb{Z}_{2^{m}})$ is $2^{m-1} \phi(p_1^{n_1}) \phi( p_2^{n_2})\\\ldots \phi(p_r^{n_r})$. By Corollary 3.3 of \cite{Plesnik1989} we get, for a connected bipartite graph with a diameter less or equal to $3$, edge connectivity of the graph is equal to the minimum degree of the graph. Therefore, $\lambda(G(\mathbb{Z}_{p_{1}^{n_{1}}} \oplus \cdots \oplus
\mathbb{Z}_{p_{r}^{n_{r}}} \oplus \mathbb{Z}_{2^{m}}))=2^{m-1}
\phi(p_1^{n_1})\ldots \phi(p_r^{n_r})$.
\end{proof}

\begin{Theorem}
Let $G(\mathbb{Z}_{p_{1}^{n_{1}}} \oplus \cdots \oplus \mathbb{Z}_{p_{r}^{n_{r}}} \oplus \mathbb{Z}_{2^{m}})$ be a unit graph and let $H$ be its incidence matrix of size $|V| \times |E|$. If all $p_1$, $p_2$, $\ldots$, $p_r$ are distinct odd primes then for any odd prime $q$, $[p_{1}^{n_{1}}\ldots
p_{r}^{n_{r}}\phi(p_{1}^{n_{1}})\ldots\phi(p_{r}^{n_{r}})2^{2(m-1)},\;2^{m}p_{1}^{n_{1}}\ldots
 p_{r}^{n_{r}}-1, 2^{m-1}\phi(p_{1}^{n_{1}})\ldots\phi(p_{r}^{n_{r}})]_q$ is the $q-$ary code generated by $H$ over the finite field $\mathbb{F}_q$. Furthermore, the dual of the above linear code is a $[p_{1}^{n_{1}}\ldots
p_{r}^{n_{r}}\phi(p_{1}^{n_{1}})\ldots\phi(p_{r}^{n_{r}})2^{2(m-1)},\;2^m p_{1}^{n_{1}}\ldots
 p_{r}^{n_{r}}
 (\phi(p_{1}^{n_{1}})\ldots\phi(p_{r}^{n_{r}})2^{m-2}-1)+1,\;4]_q-$ linear code.
\end{Theorem}

\begin{proof}
 This is straightforward. 

\end{proof}

We obtain the following corollary using the above theorem and Theorem 2.5.6 of \cite{LingXing2004}. The proof of which is omitted because it is straightforward.
\begin{Corollary}
The dual code of the linear code generated by the incidence matrix of order $|V| \times |E|$ of the unit graph $G(\mathbb{Z}_{p_{1}^{n_{1}}} \oplus \cdots \oplus \mathbb{Z}_{p_{r}^{n_{r}}} \oplus \mathbb{Z}_{2^{m}}) = (V,E)$, where $p_1$, $p_2$, $\ldots$, $p_r$ are all distinct odd prime numbers, is a $ 3 $-error-detecting code as well as a single-error-correcting code.
\end{Corollary}

\section*{\textbf{Conclusion}}
In this paper, we investigate the structural properties of the unit graph \( G(\mathbb{Z}_n) \), where  $n = p_1^{n_1} p_2^{n_2} \dots p_r^{n_r}$, with $ p_1, p_2, \dots, p_r $ being distinct prime numbers and $ n_1, n_2, \dots, n_r $ being positive integers. We establish its connectivity, determine its diameter, and analyze its edge connectivity. Furthermore, we explored the relationship between unit graphs and coding theory by constructing \( q \)-ary linear codes from their incidence matrices and explicitly determining their parameters and duals. A significant contribution of this work is the proof of two conjectures from \cite{Jain2023}, which relate to the connectivity and coding-theoretic properties of unit graphs. These results extend existing research on graph-based codes and highlight the interplay between algebraic graph theory and coding theory. The combinatorial features of unit graphs and their use in error-correcting codes, network security, and cryptography are better understood here.

\end{document}